\documentclass[10pt,journal,compsoc]{IEEEtran}
\usepackage{graphicx}
\usepackage{amsmath,amssymb,amsfonts, amsthm}
\usepackage{booktabs}
\usepackage{dirtytalk}
\usepackage[super]{nth}
\usepackage{capt-of}
\usepackage[table,xcdraw, dvipsnames]{xcolor}
\newcolumntype{?}{!{\vrule width 0.7pt}}
\newcommand\Tstrut{\rule{0pt}{2.6ex}}   
\newcommand\Bstrut{\rule[-1.2ex]{0pt}{0pt}} 
\newtheorem{prop}{Proposition}

\usepackage{hyperref}

\ifCLASSOPTIONcompsoc
  \usepackage[nocompress]{cite}
\else
  \usepackage{cite}
\fi
\ifCLASSINFOpdf
\else
\fi
\begin{document}

\title{Measuring Quadrangle Formation in Complex Networks}

\author{Mingshan Jia, Bogdan Gabrys,~\IEEEmembership{Senior Member,~IEEE,} and Katarzyna Musial
\IEEEcompsocitemizethanks{\IEEEcompsocthanksitem M. Jia, B. Gabrys and K. Musial are with the School of Computer Science, University of Technology Sydney, Ultimo NSW 2007, Australia.
E-mail: mingshan.jia@student.uts.edu.au,
\{bogdan.gabrys, katarzyna.musial-gabrys\}@uts.edu.au}
}

\markboth{Journal of \LaTeX\ Class Files,~Vol.~14, No.~8, August~2015}%
{Shell \MakeLowercase{\textit{et al.}}: Bare Demo of IEEEtran.cls for Computer Society Journals}

\IEEEtitleabstractindextext{%
\begin{abstract}
The classic clustering coefficient and the lately proposed closure coefficient quantify the formation of triangles from two different perspectives, with the focal node at the centre or at the end in an open triad respectively. As many networks are naturally rich in triangles, they become standard metrics to describe and analyse networks. However, the advantages of applying them can be limited in networks, where there are relatively few triangles but which are rich in quadrangles, such as the protein-protein interaction networks, the neural networks and the food webs. This yields for other approaches that would leverage quadrangles in our journey to better understand local structures and their meaning in different types of networks. Here we propose two quadrangle coefficients, i.e., the i-quad coefficient and the o-quad coefficient, to quantify quadrangle formation in networks, and we further extend them to weighted networks. Through experiments on $16$ networks from six different domains, we first reveal the density distribution of the two quadrangle coefficients, and then analyse their correlations with node degree. Finally, we demonstrate that at network-level, adding the average i-quad coefficient and the average o-quad coefficient leads to significant improvement in network classification, while at node-level, the i-quad and o-quad coefficients are useful features to improve link prediction.
\end{abstract}

\begin{IEEEkeywords}
clustering coefficient, closure coefficient, quadrangle coefficient, network classification, link prediction.
\end{IEEEkeywords}}

\maketitle

\IEEEdisplaynontitleabstractindextext

\IEEEpeerreviewmaketitle

\IEEEraisesectionheading{\section{Introduction}\label{sec:introduction}}

\IEEEPARstart{C}{omplex} systems across various domains, such as biology, ecology, physics and social science, can be modelled as networks that abstract the interactions between system's components \cite{barabasi2016network, newman2018networks, musial2013social}. Different from a simple grid graph or a line graph for image or text modelling respectively, the complexity of networks comes from their intricate topological structures. Therefore, the study of network structure, especially local structure, underlies a number of representative and analytical applications such as representation learning of graphs \cite{hamilton2017representation,grover2016node2vec}, node-type classification \cite{bhagat2011node, kipf2016semi}, link prediction \cite{gao2015link, kovacs2019network} and anomaly detection \cite{noble2003graph, akoglu2015graph}.

One fundamental and classic statistical metric to assess the local structure of complex networks is the \textit{local clustering coefficient} \cite{watts1998collective, fagiolo2007clustering}. It is defined as the percentage of the number of triangles formed with a focal node to the number of triangles that the focal node could form with all its neighbours. Note that the focal node here serves as the centre node in an open triad (the middle of a length-2 path). Since many of the real-world networks are triangle-rich, the clustering coefficient --- a measure of triangle formation --- has become a standard metric to describe networks. It has also been used in numerous applications such as malware detection \cite{lee2018automatic}, language learning \cite{goldstein2014influence} and structural role discovery \cite{henderson2012rolx}.

A recent study has proposed another interesting measure of triangle formation, i.e., the \textit{local closure coefficient} \cite{yin2019local}. With the focal node as the end node of an open triad (the head of a length-2 path), it is quantified as the percentage of twice the number of triangles containing the focal node to the number of all length-2 paths starting from the focal node. Specifically, the classic local clustering coefficient measures the extent to which the 1-hop neighbours of a given node connect to each other, while the local closure coefficient measures the extent to which the 2-hop neighbours of a given node connect to the given node itself. This new metric has been proven to be a useful feature in network analysis tasks such as community detection and link prediction \cite{yin2019local}.

\begin{figure}[t]
\centerline{\includegraphics[scale = 0.95]{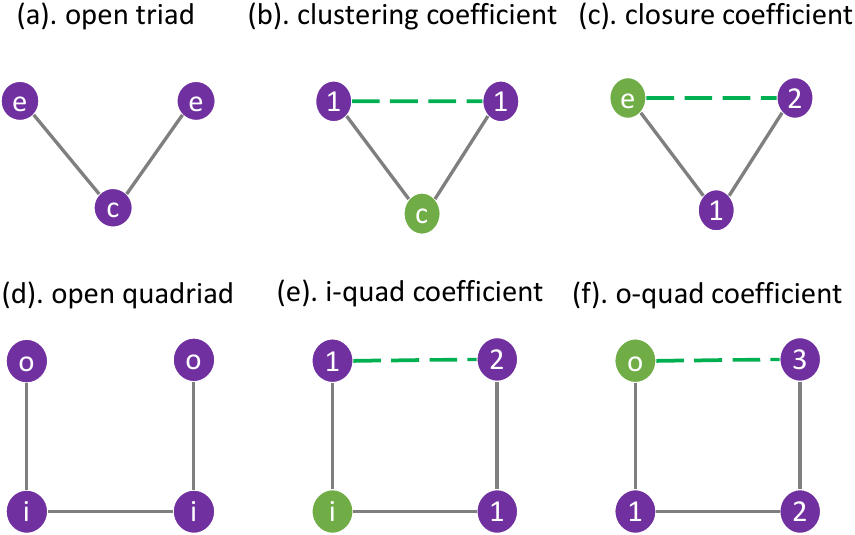}}
\caption{The i-quad coefficient and the o-quad coefficient in comparison with the clustering coefficient and the closure coefficient. Letters $c$, $e$, $i$ and $o$ denote centre node, end node, inner node and outer node respectively. Node in green colour is the focal node in each subfigure. Number on node indicates the node's distance from the focal node in the open triad or the open quadriad, which might be closed by an edge in dotted green line style.}
\label{fig:one}
\vspace{-0mm}
\end{figure}

In many types of networks, however, quadrangles appear at a much higher frequency than triangles, and thus become the most dominant motifs \cite{milo2002network}. For instance, in gene regulatory networks, logical circuits networks and neuron networks, the over-represented "bi-fan" structure (a specific directed quadrangle) serves to carry information or signals from previous units to following ones; while in food webs, the highly recurring "bi-parallel" structure (another type of directed quadrangle) describes how energy flows in an ecosystem.

In order to better describe and analyse the local structure of networks, we propose two metrics quantifying the formation of quadrangles, i.e., the \textit{i-quad coefficient} and the \textit{o-quad coefficient}. There are two definitions in that two categories of nodes --- the inner node or the outer node --- can be distinguished from the node's position in an open quadriad (also called intransitive quadriad in some works \cite{rainone2020network}). The i-quad coefficient, with the focal node functioning as the inner node of an open quadriad, measures the extent to which the focal node's 2-hop neighbours connect to its 1-hop neighbours. The o-quad coefficient, having the focal node as the outer node of an open quadriad, measures the extent to which the focal node's 3-hop neighbours connect to itself (Figure~\ref{fig:one}). 

Although the focus in this paper lies on the general unipartite networks, the proposed i-quad and o-quad coefficients provide interesting insights into bipartite networks as well. Suppose that in a recommender network where node type $x$ denotes users and node type $y$ denotes movies, an edge between $x_i$ and $y_i$ represents user $x_i$ likes movie $y_i$. Take the i-quad coefficient for instance (Figure~\ref{fig:two}a), given $x_1$, the focal user, likes movies $y_1$ and $y_2$, while $x_2$ likes $y_1$, it measures whether $x_2$ likes $y_2$. In other words, the i-quad coefficient gives the extent to which other users have a similar preference as the focal user. Likewise, for the o-quad coefficient, given $x_2$ likes $y_1$ and $y_2$, while $x_1$, the focal node, likes $y_1$, it measures whether $x_1$ likes $y_2$ (Figure~\ref{fig:two}b). That is to say, the o-quad coefficient gives the extent to which the focal user shares a similar opinion with other users. Interestingly, this explanation coincides with the idea of collaborative filtering \cite{goldberg1992using, su2009survey}.

\begin{figure}[t]
\centerline{\includegraphics[scale = 0.9]{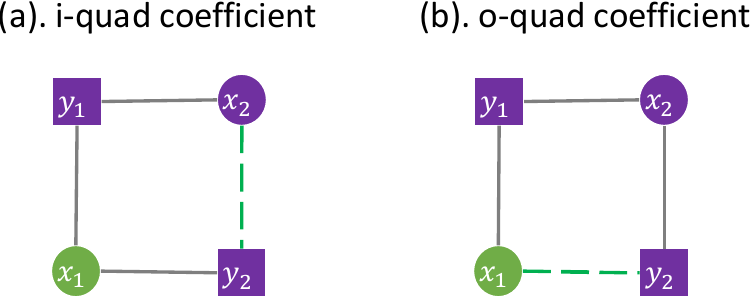}}
\caption{An example of the i-quad coefficient and the o-quad coefficient in a movie recommender network. Circle nodes represent users, and square nodes represent movies. Node $x_1$, marked in green, is the focal node. Four nodes and three solid links form an open quadriad, which if is closed by a dotted link will form a quadrangle.}
\label{fig:two}
\vspace{-0mm}
\end{figure}


In addition to the basic network structure, a deeper understanding of complex systems sometimes requires taking into account the intensity or the strength of interactions between components. This is achieved by assigning weights to links. For instance, in unipartite networks, weighted links are used to represent the frequency of contact in a communication network, or the intensity of the traffic flow in a transportation network; in bipartite networks, especially recommender networks, weights are added to indicate how much a person likes a product or how often he or she purchases it. Accordingly, we introduce the \textit{weighted i-quad coefficient} and the \textit{weighted o-quad coefficient} in order to unveil the quadrangle formation in real weighted networks.

Our empirical study on $16$ real-world networks from six domains has revealed several basic and interesting properties of the two proposed coefficients.
First, we find that in most types of networks, the average o-quad coefficient is smaller than the average i-quad coefficient, which is also demonstrated through their cumulative density distributions. Secondly, we show that the o-quad coefficient has a strong positive correlation with node degree, whereas the correlation between the i-quad coefficient and node degree is very weak. We then provide a theoretical justification of this phenomenon under the configuration model. 

Last but not least, we illustrate how the proposed quadrangle coefficients can be powerful features for network analysis and inference tasks. In a network classification task, we show that different types of real-world networks are significantly better clustered by adding the two quadrangle coefficients. Furthermore, in a link prediction task, we also show that the i-quad and o-quad coefficients can be used as effective predictors to improve the performance, especially in food webs, protein-protein interaction networks and infrastructure networks.

To sum up, in order to measure the formation of quadrangles in networks, we propose the i-quad coefficient and the o-quad coefficient, based on the inner node and the outer node of an open quadriad respectively. We further extend them to weighted networks. Through extensive experiments on real-world networks, we show not only the intrinsic properties of the two coefficients, but also investigate how they can be utilised in common network analysis task and machine learning tasks. The remainder  of this paper is organised as follows. Section~\ref{sec: background} introduces notations and background knowledge of clustering coefficient and closure coefficient. Section~\ref{sec: quadrangle_coefs} presents and exemplifies the proposed quadrangle coefficients, whereas Section~\ref{sec: evaluation} provides details of the evaluation, including the datasets, experiment setups, performance measures, experiment results and our findings. Section~\ref{sec: related_works} briefly contemplates the related works, and finally we conclude this paper in Section~\ref{sec: conclusion}.

\section{Background and Motivating Example} \label{sec: background}
This section first introduces the basic concepts such as the classic clustering coefficient and the recently proposed closure coefficient. We then illustrate how these coefficients are calculated in the case of a small-scale network that serves as an example.

\subsection{Clustering Coefficient} \label{sec:2.1}
The clustering coefficient, or more specifically the local clustering coefficient, was originally proposed in order to measure the cliquishness of a neighbourhood in networks \cite{watts1998collective}. It has since become one of the most commonly used metrics for network structure, together with such measures as degree distribution, path length, connected components, etc. 
Let $G = (V,E)$ be an undirected graph on a node set $V$ (the number of nodes is $|V|$) and an edge set $E$ (the number of edges is $m$), without self-loops and multiple edges. We denote the set of neighbours of node $i$ as $N(i)$, and thus the degree of node $i$, denoted as $d_{i}$, equals to $|N(i)|$. An open triad is a directionless length-2 path. For example, in an open triad $ijk$, where an edge connects node $i$ and $j$, and another edge connects node $j$ and $k$, we do not distinguish between path $i\rightarrow j \rightarrow k$ and path $k\rightarrow j \rightarrow i$.

For any node $i \in V$, its \textit{local clustering coefficient}, denoted $C(i)$, is defined as the number of triangles containing node $i$ (denoted $T(i)$), divided by the number of open triads with $i$ as the centre node (denoted $OTC(i)$):
\begin{equation}
C(i) =\frac{T(i)}{O T C(i)}=\frac{\frac{1}{2} \sum_{j \in N(i)}|N(i) \cap N(j)|}{\frac{1}{2} d_{i}\left(d_{i}-1\right)}.
\end{equation}
In other words, it is the fraction of open triads, where the focal node serves as the centre node, that actually form triangles. By definition, $C(i) \in [0,1]$.

In order to get a network-level measurement, the \textit{average clustering coefficient} is introduced by averaging the local clustering coefficient over all nodes (an undefined local clustering coefficient is treated as zero): 
\begin{equation} \label{eq: avg_clu}
\overline{C}=\frac{1}{|V|} \sum_{i \in V} C(i).
\end{equation}

An alternative way to measure clustering at the network-level is the \textit{global clustering coefficient} \cite{newman2001random}, which is defined as the fraction of open triads that form triangles in the entire network:
\begin{equation}
\label{eqn_gcc}
C=\frac{\sum_{i \in V} \sum_{j \in N(i)}|N(i) \cap N(j)|}{\sum_{i \in V}d_{i}\left(d_{i}-1\right)}.
\end{equation}
Note that the global clustering coefficient is not equivalent to the average clustering coefficient. In Equation~\ref{eqn_gcc}, we calculate the number of triangles in the entire network, then divided by the number of open triads across the network. Since a node with high degree forms more open triads and also tends to form more triangles, the global clustering coefficient thus puts more weight on hub nodes. On the contrary, in Equation~\ref{eq: avg_clu}, we first calculate the sum of local clustering coefficient of each node, then average over the number of nodes, which gives equal weight on each node. 

\subsection{Closure Coefficient}
Different from the ordinary centre node based perspective in the clustering coefficient, another interesting measure of triangle formation, i.e., the local closure coefficient, has recently been proposed \cite{yin2019local}. The focal node in the closure coefficient serves as the end node of an open triad. As Yin et al.\cite{yin2019local} has revealed, this subtle difference in measurement leads to very different properties from those of the clustering coefficient. 

Adopting the notations of Section~\ref{sec:2.1}, the local closure coefficient of node $i$, denoted $E(i)$, is defined as twice the number of triangles formed with $i$, divided by the number of open triads with $i$ as the end node. (denoted $OTE(i)$):
\begin{equation} \label{eqn_lcc}
    E(i)=\frac{2T(i)}{OTE(i)}=\frac{\sum_{j \in N(i)}|N(i) \cap N(j)|}{\sum_{j \in N(i)}(d_j-1)}.
\end{equation}
In other words, it is the fraction of open triads, where the focal node serves as the end node, that actually form triangles. $T(i)$ is multiplied by two for the reason that each triangle contains two open triads with $i$ as the end node. When a triangle is actually formed, the focal node can be viewed as the centre node in one open triad or as the end node in two open triads (Figure~\ref{fig:open_triad}). Obviously, $E(i) \in [0,1]$. 

\begin{figure}[t]
\centerline{\includegraphics[scale = 1]{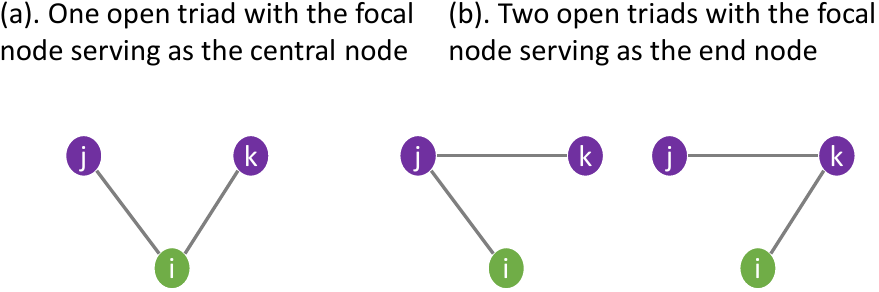}}
\caption{Two types of open triads in triangle formation. Among three nodes $i$, $j$ and $k$,  node $i$, painted in green, is the focal node.}
\label{fig:open_triad}
\vspace{-0mm}
\end{figure}

At the network-level, the \textit{average closure coefficient} is then defined as the mean of the local closure coefficient over all nodes (an undefined local closure coefficient is treated as zero):
\begin{equation}
\overline{E}=\frac{1}{|V|} \sum_{i \in V} E(i). 
\end{equation}

Analogous to the global clustering coefficient (Equation~\ref{eqn_gcc}), the \textit{global closure coefficient}, denoted $E$, is defined as:
\begin{equation}
\label{eqn_gce}
E=\frac{\sum_{i \in V} \sum_{j \in N(i)}|N(i) \cap N(j)|}{{\sum_{i \in V} \sum_{j \in N(i)}(d_j-1)}}.
\end{equation}

The global closure coefficient (Equation~\ref{eqn_gce}) is actually equivalent to the global clustering coefficient (Equation~\ref{eqn_gcc}), as globally the difference of the position of the focal node will not surface.

\subsection{A motivating example}
 
\begin{figure*}[t]
\vspace{1mm}
\centerline{\includegraphics[scale = 1]{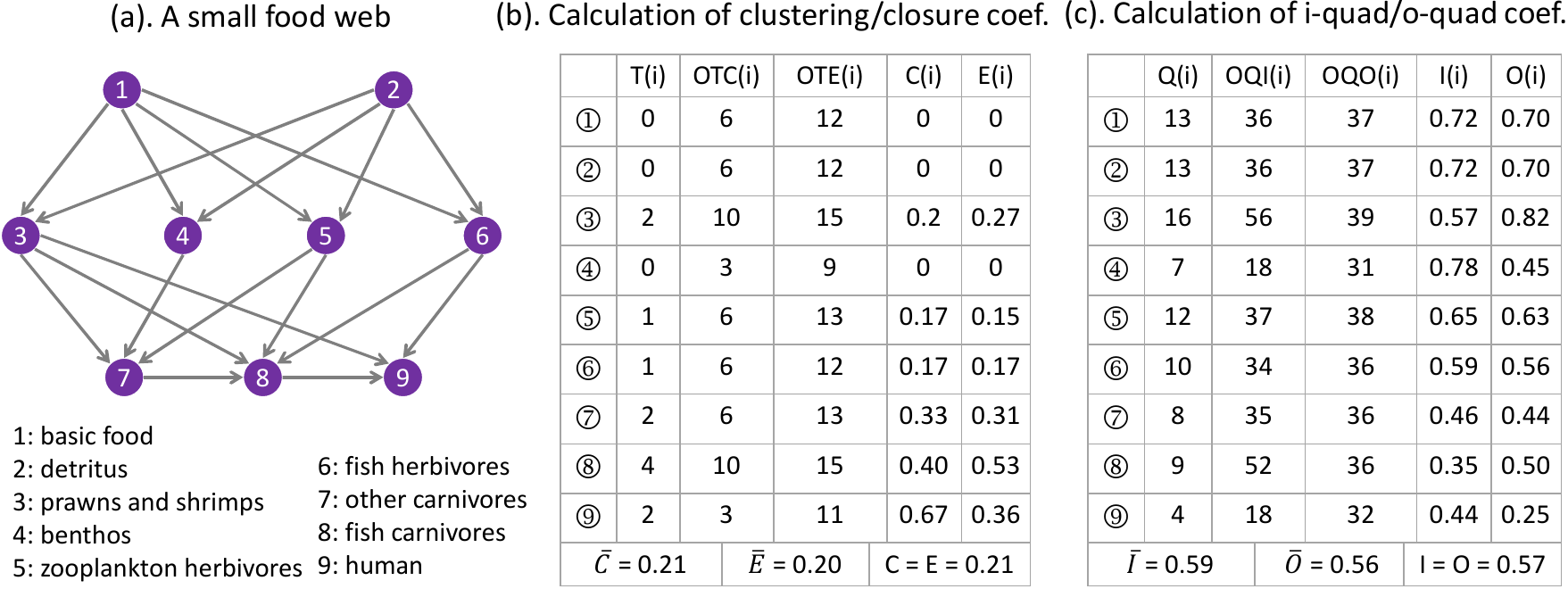}}
\caption{A motivating example.}
\label{fig:ex}
\vspace{-0mm}
\end{figure*}

We illustrate how the two coefficients of triangle formation are calculated via a small yet real network.
Figure~\ref{fig:ex}a shows a simplified food web of the backwaters of Kerala, India \cite{qasim1970some}. It is composed of 9 nodes and 18 edges. Each node represents a species and each edge represents the flow of food energy from one species to another.  

Figure~\ref{fig:ex}b gives a detailed table of the number of triangles $T(i)$, the number of centre-node-based open triads $OTC(i)$, the number of end-node-based open triads $OTE(i)$, the local clustering coefficient $C(i)$ and the local closure coefficient $E(i)$ for each node. Also, the last row gives the average clustering coefficient, the average closure coefficient and the global clustering/closure coefficient, all of which are around $0.20$.

Different from some triangle-rich networks, we find many more quadrangles than triangles in this food web (23 versus 4), which motivates us to propose measuring quadrangle formation instead. In the next section, new measures to quantify information about quadrangles in complex networks are proposed, and we show how we can leverage the fact that some networks are quadrangle and not triangle rich.
\section{Two Quadrangle Coefficients} \label{sec: quadrangle_coefs}
The clustering coefficient and the closure coefficient provide us two ways of measuring triangle formation. In some networks however, we care more about the formation of quadrangles. Also, triangles do not exist in bipartite networks and the most basic enclosed structure in this representation of networks is quadrangle. In this section, we first propose two coefficients measuring quadrangle formation, based on two different positions of the focal node in an open quadriad. Then, we further extend them to weighted networks. 

\subsection{I-quad coefficient}
Recall that an open quadriad is a directionless length-3 path (Figure~\ref{fig:one}d). In an open quadriad $ijkl$, for instance, where three edges exist between node pairs $(i, j)$, $(j, k)$ and $(k, l)$, we name nodes $j$ and $k$ as inner nodes. In contrast, nodes $i$ and $l$ are outer nodes. Obviously, an inner node has a degree of two, and an outer node has a degree of one. Further, an open quadriad with the focal node acting as the inner node is called inner-node-based open quadriad of that node; an open quadriad with the focal node acting as the outer node is named outer-node-based open quadriad of that node.

In comparison to the definition of clustering coefficient in measuring triangle formation, we propose the i-quad coefficient for measuring quadrangle formation. It is quantified as the fraction of inner-node-based open quadriads that actually form quadrangles. Concretely, the  \textbf{\textit{i-quad coefficient}} of node $i$, denoted $I(i)$, is defined as twice the number of quadrangles formed with $i$ (denoted as $Q(i)$), divided by the number of open quadriads with $i$ as the inner node (denoted as $OQI(i)$):
\begin{equation} \label{eqn: i-quad}
\begin{split}
I (i) & 
= \frac{2 Q(i)}{OQI(i)} \\
& =  \frac{\sum_{j \in N(i)} \sum_{k \in(N(j)-i)}|N(k) \cap N(i)-j|}{\sum_{j \in N(i)} \sum_{k \in (N(j)-i)}|N(i)-j-k|}.
\end{split}
\end{equation}
In the above equation, $j$ is in $i$'s neighbour set, and $k$ is in $j$'s neighbour set excluding $i$. $Q(i)$ is multiplied by two because each quadrangle can be viewed as constructed from two open quadriads with $i$ as the inner node. By definition, it is obvious that $I(i) \in [0,1]$.

Then, we define the \textbf{\textit{average i-quad coefficient}} at the network-level, as the mean of the i-quad coefficient over all nodes (undefined ones are treated as zeros):
\begin{equation}
\overline{I}=\frac{1}{|V|} \sum_{i \in V} I(i). 
\end{equation}
In the case of a random network where each pair of nodes is connected with a probability $p$, the expected value of the average i-quad coefficient is also $p$, i.e., $\mathop{\mathbb{E}}[\overline{I}] = p$.

An alternative way of measuring quadrangle formation at the network-level is the \textbf{\textit{global i-quad coefficient}}, which is defined as the fraction of inner-node-based open quadriads that form quadrangles in the entire network:
\begin{equation}
\label{eqn_giq}
I = \frac{\sum_{i \in V} \sum_{j \in N(i)} \sum_{k \in(N(j)-i)}|N(k) \cap N(i)-j|}{{\sum_{i \in V} \sum_{j \in N(i)} \sum_{k \in (N(j)-i)}|N(i)-j-k|}}.
\end{equation}
The numerator of the above equation can be viewed as eight times the number of quadrangles in the entire network (each node of a quadrangle contributes two counts), then divided by twice the number of open quadriads with each node acting as the inner node.

Although both the average i-quad coefficient and the global i-quad coefficient can be used as metrics to describe quadrangle formation in the entire network, they are calculated differently. The average i-quad coefficient adds up the i-quad coefficient of every node then divides it by the number of nodes, giving each node equal weight. In contrast, the global i-quad coefficient gives nodes that form numerous quadrangles more weight, by first totalling the numerator of the i-quad coefficient then dividing it by the sum of the denominator of the i-quad coefficient.

\subsection{O-quad coefficient}
\begin{figure}[t]
\centerline{\includegraphics[scale = 0.95]{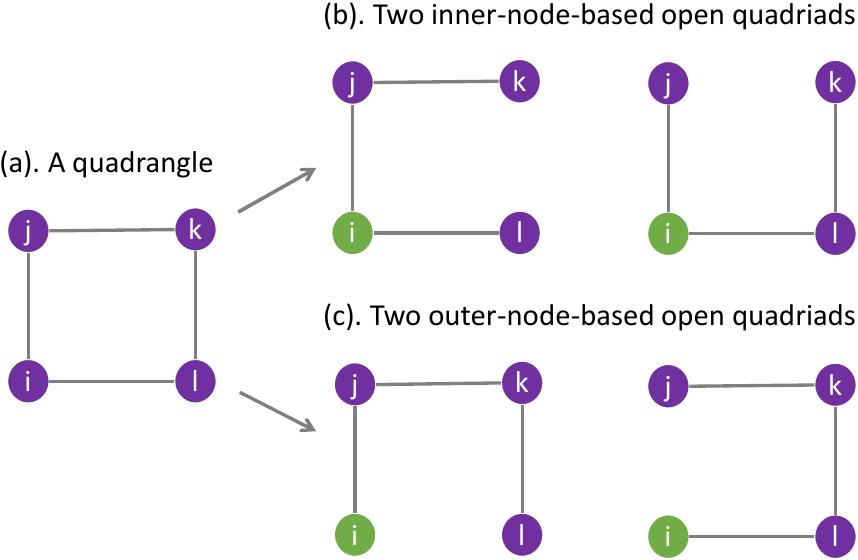}}
\caption{Two types of open quadriads in a quadrangle. Node $i$, depicted in green, is the focal node, among four nodes $i$, $j$, $k$ and $l$.}
\label{fig:four}
\vspace{-0mm}
\end{figure}

Inspired by the closure coefficient in measuring triangle formation, we move the focal node from the inner node to the outer node of an open quadriad, thus proposing the o-quad coefficient in order to measure the formation of quadrangle from a different perspective. 

The significance of introducing the o-quad coefficient is twofold. 
First, the o-quad coefficient takes into account length-$3$ paths emanating from the focal node, and therefore has a larger scope of the network structure.
Second, when a quadrangle is formed, the closing edge (the edge that closes the outer-node-based open quadriad) is incident to the focal node. This leads to some special properties, comparing to the i-quad coefficient where the closing edge is not incident to the focal node. We show in Section~\ref{sec: evaluation} that the cumulative distribution curve of the o-quad coefficient is above that of the i-quad coefficient, and that the o-quad coefficient tends to increase with node degree. 

In a similar way, the \textbf{\textit{o-quad coefficient}} of node $i$, denoted as $O(i)$,  is defined as the fraction of open quadriads with $i$ as the outer node that are closed:
\begin{equation} \label{eqn: o-quad}
\begin{split}
O (i) & 
= \frac{2 Q(i)}{OQO(i)} \\
& = \frac{\sum_{j \in N(i)} \sum_{k \in(N(j)-i)}|N(k) \cap N(i)-j|}{\sum_{j \in N(i)} \sum_{k \in (N(j)-i)}|N(k)-j-i|},
\end{split}
\end{equation}
where $OQO(i)$ is the number of outer-node-based open quadriads of node $i$, and $Q(i)$ is the number of quadrangles containing $i$. $Q(i)$ is multiplied by two because each quadrangle contains two open quadriads with $i$ as the outer node. In a quadrangle, the focal node can serve as the inner node in two open quadriads or as the outer node in another two open quadriads (Figure~\ref{fig:four}). Obviously, $O(i) \in [0,1]$.

In order to measure at the network level, the \textbf{\textit{average o-quad coefficient}} is defined by averaging the o-quad coefficient over all nodes (an undefined o-quad coefficient is treated as zero):
\begin{equation}
\overline{O}=\frac{1}{|V|} \sum_{i \in V} O(i). 
\end{equation}

Analogous to the global i-quad coefficient, the \textbf{\textit{global o-quad coefficient}} can be defined as the fraction of outer-node-based open quadriads that form quadrangles in the entire network:
\begin{equation}
\label{eqn_goq}
O = \frac{\sum_{i \in V} \sum_{j \in N(i)} \sum_{k \in(N(j)-i)}|N(k) \cap N(i)-j|}{{\sum_{i \in V} \sum_{j \in N(i)} \sum_{k \in (N(j)-i)}|N(k)-j-i|}}.
\end{equation}
As the equivalence between the global clustering coefficient and the global closure coefficient, this definition of global o-quad coefficient is actually not different from the global i-quad coefficient (Equation~\ref{eqn_giq}) since globally the difference of the position of the focal node will not arise.

Revisiting the motivating example, Figure~\ref{fig:ex}c gives a detailed table of the number of quadrangles $Q(i)$, the number of inner-node-based open quadriads $OQI(i)$ and the number of outer-node-based open quadriads $OQO(i)$ of each node, based on which the i-quad coefficient $I(i)$ and the o-quad coefficient $O(i)$ are calculated. Also, the last row of this table gives the three network-level measures, i.e., the average i-quad coefficient, the average o-quad coefficient and the global i-quad/o-quad coefficient, which are more than $2.5$ times larger than those metrics measuring triangles formation.

\subsection{Quadrangle coefficients in weighted networks}

Until now, the discussion has been focused on binary networks, where the value of each link is either one or zero. In many networks, however, we need a more accurate representation of the relationships between nodes, such as the frequency of contact in a communication network, or the rating of a product given by a consumer in a recommender network, etc. This kind of information is usually expressed as a strength of the relationship and we use weighted networks to represent it. Therefore, we are interested in extending the two quadrangle coefficients to networks that allow for weights of the relationships.

Several versions of weighted clustering coefficient have been proposed in order to measure triangle formation in weighted networks \cite{barrat2004architecture, onnela2005intensity, zhang2005general, saramaki2007generalizations}. For example, Onnela et al. \cite{onnela2005intensity} proposed to sum over the geometric averages of the three weights in formed triangles, divided by the number of potential triangles. Alternatively, Zhang and Horvath. \cite{zhang2005general} chose to sum simply over the products of the three weights in formed triangles, divided by the total of products of the two weights of all open triads, implying the triadic closing edges taking the maximum weight. 

Adopting a strategy similar to the one proposed by Zhang and Horvath \cite{zhang2005general}, we introduce the weighted i-quad coefficient and the weighted o-quad coefficient to measure quadrangles formation in weighted networks. Let $G^\mathcal{W} = (V, E)$ be a weighted graph without self-loops and multiple edges. The weight of a link between any node $i$ and $j$ is denoted $w_{i j}$ ($w_{i j} \in [0,1]$ after normalisation by the maximum weight). For any node $i \in V$, the \textit{\textbf{weighted i-quad coefficient}}, denoted as $I^\mathcal{W}(i)$, and the \textbf{\textit{weighted o-quad coefficient}}, denoted as $O^\mathcal{W}(i)$, are defined as:

\begin{equation}
I^\mathcal{W}(i) = \frac{\sum\limits_{j \in N(i)} \sum\limits_{k \in(N(j)-i)} \sum\limits_{l \in(N(i) \cap N(k)-j)} w_{i j} w_{j k} w_{i l} w_{l k}}{\sum\limits_{j \in N(i)} \sum\limits_{k \in(N(j)-i)} \sum\limits_{l \in(N(i)-j-k)} w_{i j}  w_{j k} w_{i l}}, 
\end{equation}

\begin{equation}
O^\mathcal{W}(i) = \frac{\sum\limits_{j \in N(i)} \sum\limits_{k \in(N(j)-i)} \sum\limits_{l \in(N(i) \cap N(k)-j)} w_{i j} w_{j k} w_{i l} w_{l k}}{\sum\limits_{j \in N(i)} \sum\limits_{k \in(N(j)-i)} \sum\limits_{l \in(N(k)-j-i)} w_{i j}  w_{j k} w_{k l}}. 
\end{equation}

When the graph becomes binary (unweighted), i.e., $w_{i j} = 1$, the above two weighted quadrangle coefficients degrade to their unweighted versions (Equation~\ref{eqn: i-quad} and Equation~\ref{eqn: o-quad}). The average weighted i-quad coefficient and the average weighted o-quad coefficient are then defined respectively as: $\overline{I^\mathcal{W}}=\frac{1}{|V|} \sum_{i \in V} I^\mathcal{W}(i)$, $\overline{O^\mathcal{W}}=\frac{1}{|V|} \sum_{i \in V} O^\mathcal{W}(i)$.

We can see from Figure~\ref{fig:weighted_quad} that in different weighted networks, the correlation of i-quad coefficient and weighted i-quad coefficient (and the correlation of o-quad coefficient and weighted o-quad coefficient) is also different. In other words, when weights are considered in calculating quadrangle coefficients, the weighted i-quad coefficient and the weighted o-quad coefficient capture different information compared to their unweighted counterparts. 

\begin{figure}[t]
\centerline{\includegraphics[scale = 0.24]{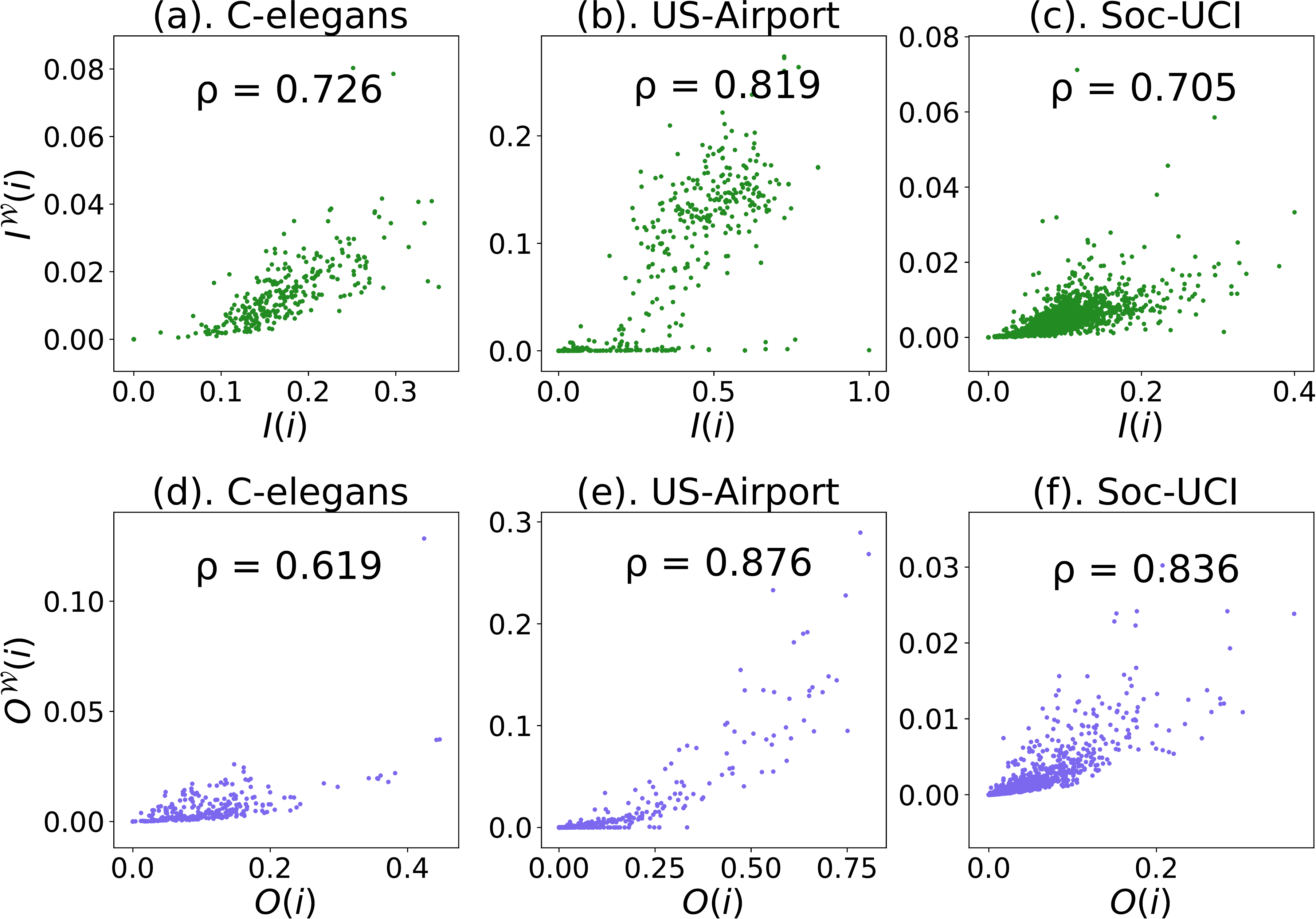}}
\caption{Correlation of quadrangle coefficients and weighted quadrangle coefficients in three different networks. First row is the correlation of i-quad coefficient $I(i)$ and weighted i-quad coefficient $I^\mathcal{W}(i)$, second row is the correlation of o-quad coefficient $O(i)$ and weighted o-quad coefficient $O^\mathcal{W}(i)$. The weighted networks are: (1) the neural network of the Caenorhabditis elegans worm \cite{watts1998collective}; (2) the network of the $500$ busiest commercial airports in the United States\cite{colizza2007reaction}; (3) the social network of online community for students at University of California, Irvine\cite{opsahl2009clustering}.}
\label{fig:weighted_quad}
\vspace{-0mm}
\end{figure}

\subsection{Computational cost}
At the end of this section, we give a brief discussion about the computational efficiency of the above mentioned metrics. From Equation~\ref{eqn: i-quad} and Equation~\ref{eqn: o-quad}, we can see that to compute the i-quad coefficient or the o-quad coefficient for a single node, the cost is $O({\langle k\rangle}^3)$, where $\langle k\rangle$ is the average degree of the network. Therefore, the cost for computing the two coefficients for every node in a network is $O({|V| \cdot \langle k\rangle}^3)$. This might seem expensive. Fortunately, in most real-world networks, $\langle k\rangle$ is small, and therefore the computation of these proposed metrics is relatively fast in large networks.

\begin{table*}[th]
\vspace{0mm}
\centering
\caption{Statistics of datasets, showing the number of nodes ($|V|$), the number of edges ($|E|$), the average degree ($\langle k\rangle$), the average clustering coefficient ($\overline{C}$), the average closure coefficient ($\overline{E}$), the average i-quad coefficient ($\overline{I}$) and the average o-quad coefficient ($\overline{O}$). In order to facilitate comparison, the last four columns give the quotient of $\overline{C}$ and $\overline{E}$, the quotient of $\overline{I}$ and $\overline{O}$, the quotient of $\overline{I}$ and $\overline{C}$, and the quotient of $\overline{O}$ and $\overline{E}$ respectively. Datasets having timestamps on edge creation are superscripted by ($\tau$).}
\label{tab:dataset}
\setlength{\tabcolsep}{9pt} 
\def\arraystretch{1.15}
\begin{tabular}{lrrrcccc@{\hskip 10pt}?@{\hskip 10pt}cccc}
\toprule
Network      & $|V|$   & $|E|$  & $\langle k\rangle$  & $\overline{C}$  & $\overline{E}$ &  $\overline{I}$  & $\overline{O}$ & $\overline{C}/\overline{E}$ & $\overline{I}/\overline{O}$ & $\overline{I}/\overline{C}$ &  $\overline{O}/\overline{E}$ \Bstrut\\
\hline
\textsc{FW-FloridaDry}   & 128 & 2,106 & 32.91 & 0.335 & 0.261 & 0.428 & 0.353 & 1.280 & 1.213 & 1.280  & 1.351   \Tstrut\\
\textsc{FW-LittleRock}   & 183 & 2,452  & 26.80 & 0.323 & 0.208 & 0.550 & 0.339 & 1.553 & 1.622 & 1.704  & 1.631 \Bstrut\\
\hline
\textsc{Soc-EmailEu$^{\tau}$}  & 986  & 16,064  & 32.58 & 0.407 & 0.153 & 0.231 & 0.102 & 2.659 & 2.267 & 0.568  & 0.667 \Tstrut\\
\textsc{Soc-ClgMsg$^{\tau}$}  & 1,899  & 13,838 & 14.57 & 0.109 & 0.022 & 0.081 & 0.029 & 5.082 & 2.806 & 0.744 & 1.347  \\
\textsc{Soc-BTCAlpha$^{\tau}$}  & 3,783  & 14,124 & 7.47 & 0.177 & 0.020 & 0.058 & 0.013 & 8.937 & 4.448 & 0.326 & 0.655   \\ 
\textsc{Soc-TwitchFr}  & 6,549  & 113K & 34.41 & 0.222 & 0.029 & 0.109 & 0.034  & 7.557 & 3.202 & 0.493 & 1.163  \Bstrut \\
\hline
\textsc{PPI-Stelzl}   & 1,706 & 3,191 & 3.74 & 0.006 & 0.002 & 0.038 & 0.021 & 3.827 & 1.806 & 6.332 & 13.416 \Tstrut\\
\textsc{PPI-Figeys}   & 2,239 & 6,432 & 5.75 & 0.040 & 0.005 & 0.082 & 0.043 & 7.321 & 1.908 & 2.064 & 7.918\\
\textsc{PPI-Vidal}   & 3,133 & 6,726  & 4.29 & 0.064 & 0.025 & 0.040 & 0.018 & 2.531 & 2.291 & 0.632 & 0.698 \\
\textsc{PPI-IntAct}   & 8077 & 26,085  & 6.46 & 0.083 & 0.016 & 0.063 & 0.021 & 5.101 & 2.993 & 0.750 & 1.278 \Bstrut\\
\hline
\textsc{Cit-DBLP$^{\tau}$}   & 12,590 & 49,651 & 7.89 & 0.117 & 0.026 & 0.060 & 0.014 & 4.529 & 4.175 & 0.510 & 0.553 \Tstrut\\
\textsc{Cit-Cora}   & 23,166 & 89,157 & 7.70 & 0.266 & 0.100 & 0.107 & 0.047 & 2.667 & 2.285 & 0.402 & 0.469 \Bstrut\\
\hline
\textsc{Rd-NewYork}   & 264K & 365K & 2.76  & 0.021 & 0.021 & 0.068 & 0.069 & 1.012 & 0.990 & 3.291 & 3.365 \Tstrut \\
\textsc{Rd-BayArea}   & 321K & 397K & 2.47 & 0.017 & 0.016 & 0.038 & 0.038 & 1.020 & 0.992 & 2.284 & 2.350  \Bstrut\\
\hline
\textsc{QA-MathOvfl.$^{\tau}$} & 21,688  & 88,956 & 8.20 & 0.094 & 0.005 & 0.031 & 0.004 & 17.956 & 7.305 & 0.333 & 0.817 \Tstrut\\
\textsc{QA-AskUbuntu$^{\tau}$} & 138K  & 262K  & 3.81 & 0.015 & 5e-4 & 0.004 & 5e-4 & 31.708 & 7.867 & 0.243  & 0.981\\
\bottomrule
\end{tabular}
\vspace{-0mm}
\end{table*}

\section{Experiments and Analysis} \label{sec: evaluation}

In this section, we analyse the proposed quadrangle coefficients on different types of real-world networks and demonstrate their usage in some common applications\footnote{Our code is available at \url{https://github.com/MingshanJia/explore-local-structure}.}. 

\subsection{Quadrangle coefficients in real-world networks}
\textbf{\textit{Datasets}.} We run experiments on 16 networks of six categories: \begin{enumerate}
    \item Food webs. \textsc{FloridaDry}\cite{ulanowicz1999network, kunegis2013konect} and \textsc{LittleRock} \cite{martinez1991artifacts}: energy transfer relationships collected from the cypress wetlands of South Florida and the Little Rock Lake of Wisconsin. Nodes represent species and an edge denotes that one species feeds on another (edge direction and weight are ignored).
    \item Social networks. \textsc{EmailEu}\cite{paranjape2017motifs, leskovec2016snap}: a temporal email network from a European research institution (a temporal edge denotes that an email is exchanged between two persons at a certain time); \textsc{ClgMsg}\cite{panzarasa2009patterns}: temporal online message interactions between UCIrvine college students (a temporal edge means that a message is exchanged between two students at a certain time); \textsc{BTCAlpha} \cite{kumar2018rev2}: a temporal who-trusts-whom network of users on a Bitcoin trading platform Bitcoin Alpha (edge direction and weight are ignored); \textsc{TwitchFr} \cite{rozemberczki2019multi}: a network of gamers who stream in French, where nodes are the users and edges are mutual friendships between them.
    \item Protein-protein interaction networks. \textsc{Stelzl}\cite{stelzl2005human}, \textsc{Figeys}\cite{ewing2007large}, \textsc{Vidal}\cite{rual2005towards} and \textsc{IntAct}\cite{orchard2014mintact}: four networks of interactions between proteins in Homo sapiens. Nodes represent proteins and an edge denotes the physical contact between two proteins in the cell.
    \item Citation networks. \textsc{DBLP}\cite{ley2002dblp} and \textsc{Cora}\cite{vsubelj2013model}: two academic publication citation networks. \textsc{DBLP} contains temporal information on edges. Nodes represent papers, and an edge means that one paper cites another paper (direction is ignored). 
    \item Infrastructure networks. \textsc{Rd-NewYork} and \textsc{Rd-BayArea}\cite{kunegis2013konect}: two road networks for New York City and San Francisco Bay Area. Nodes represent intersections and endpoints, and the roads connecting them are represented by edges.
    \item Q$\&$A networks. \textsc{MathOvfl.} and \textsc{AskUbuntu}\cite{paranjape2017motifs}: two temporal Q$\&$A networks derived from Stack Exchange. Nodes represent users, and a temporal edge means that one user answers another user's question at a certain time (edge direction is ignored). 
\end{enumerate}

\noindent\textbf{\textit{Observations}.} 
Table~\ref{tab:dataset} lists some key statistics including the proposed coefficients of these networks. We observe that in most types of networks (except road networks), the average o-quad coefficient is smaller than the average i-quad coefficient. That is to say, for the majority of nodes in these types of networks, fewer quadrangles are built from the outer-node-based open quadriads, compared to the number of quadrangles constructed from the inner-node-based open quadriads. This phenomenon is better revealed through the cumulative distribution function (Figure~\ref{fig:cdf}): the CDF curve of the o-quad coefficient is above that of the i-quad coefficient when the coefficient value is small (except in \textsc{Rd-NewYork}). 

We can also observe that in all food webs, two PPI networks (\textsc{PPI-Stelzl} and \textsc{PPI-Figeys}) and all road networks, the average i-quad coefficient is larger than the average clustering coefficient ($\overline{I}>\overline{C}$); and the average o-quad coefficient is larger than the average closure coefficient ($\overline{O}>\overline{E}$). In other words, these networks are more inclined to form quadrangles than to form triangles, which leads us to the following experiments.

\begin{figure}[t]
\centerline{\includegraphics[scale = 0.47]{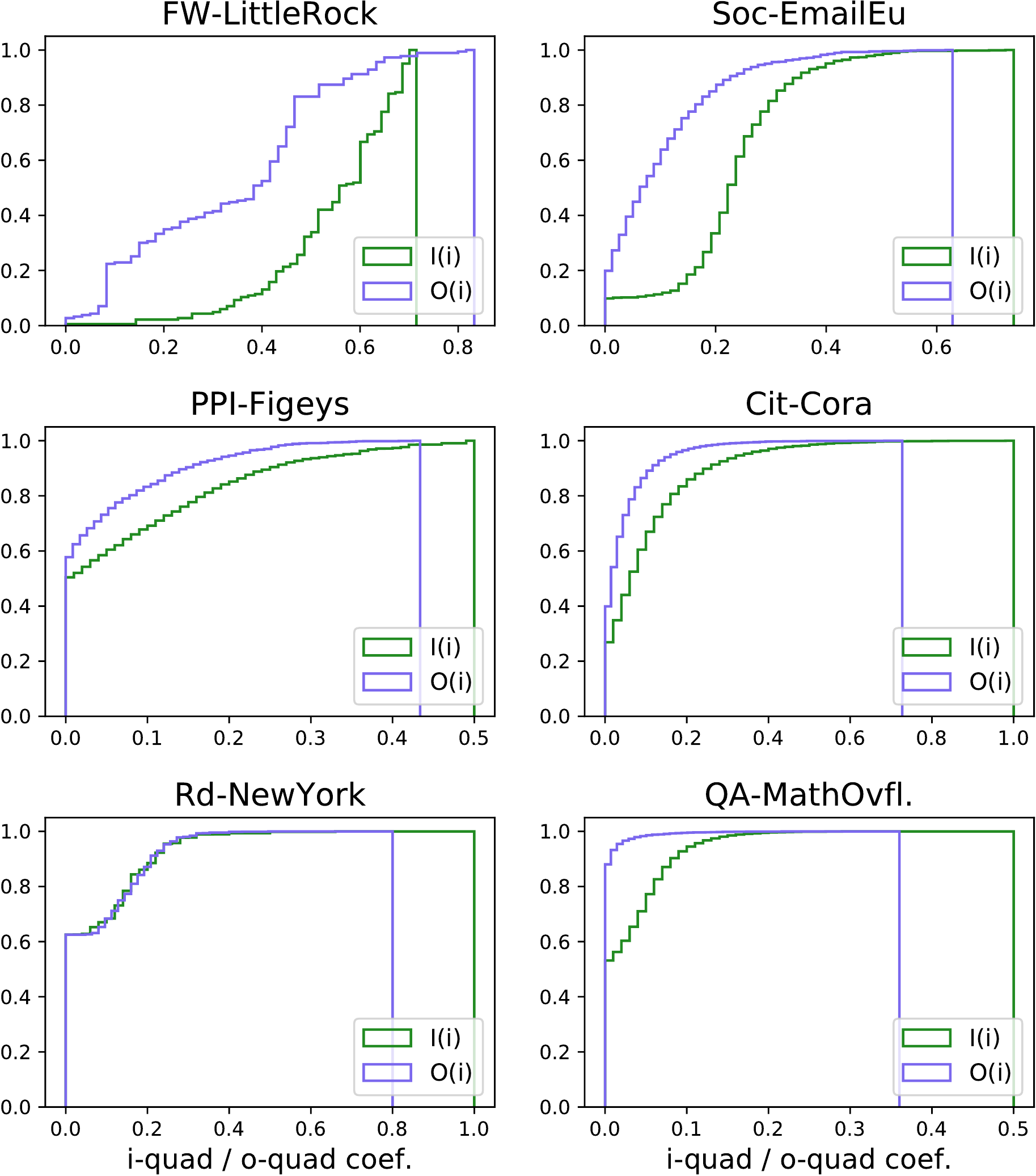}}
\caption{Cumulative distribution curve of the i-quad coefficient $I(i)$ (in green colour) and the o-quad coefficient $O(i)$ (in purple colour) in six real-world networks of different types.}
\label{fig:cdf}
\vspace{-0mm}
\end{figure}

\subsection{Correlation with node degree}

Since node degree is one of the most important and widely used concepts in network science, we study how the two quadrangle coefficients vary with it. We start by conducting an empirical analysis in real networks, followed by a theoretical justification under the degree-preserving random graph model.

We choose one network in each category and plot the correlation of quadrangle coefficients and degree (Figure~\ref{fig:corr_degree}). We observe a strong positive correlation between the o-quad coefficient and the node degree: 
the average o-quad coefficient is small among nodes with small degree and becomes larger as the average node degree increases.
In contrast, the correlation between the i-quad coefficient and the degree is weak: the average i-quad coefficient is large (compared to the average o-quad coefficient) when the average node degree is small and does not change too much as the average degree increases. Since most real-world networks are scale-free and exhibit heavy-tailed degree distribution, it also explains why the average i-quad coefficient is bigger than the average o-quad coefficient in most networks studied in our work (Table~\ref{tab:dataset}).

To better understand the correlation between the quadrangle coefficients and the node degree, we give a theoretical explanation under the configuration model \cite{fosdick2018configuring}. Constrained by a given degree sequence, the configuration model generates a network by placing edges between nodes uniformly at random. This can be achieved through a stub-matching process, in which the probability of forming an edge between node $i$ and node $j$ equals $ d_i \cdot d_j / 2m$ (assuming $d_{i}^{2} \leqslant 2 m$ for all $i$). Now we give the following proposition.

\begin{figure}[t]
\centerline{\includegraphics[scale = 0.47]{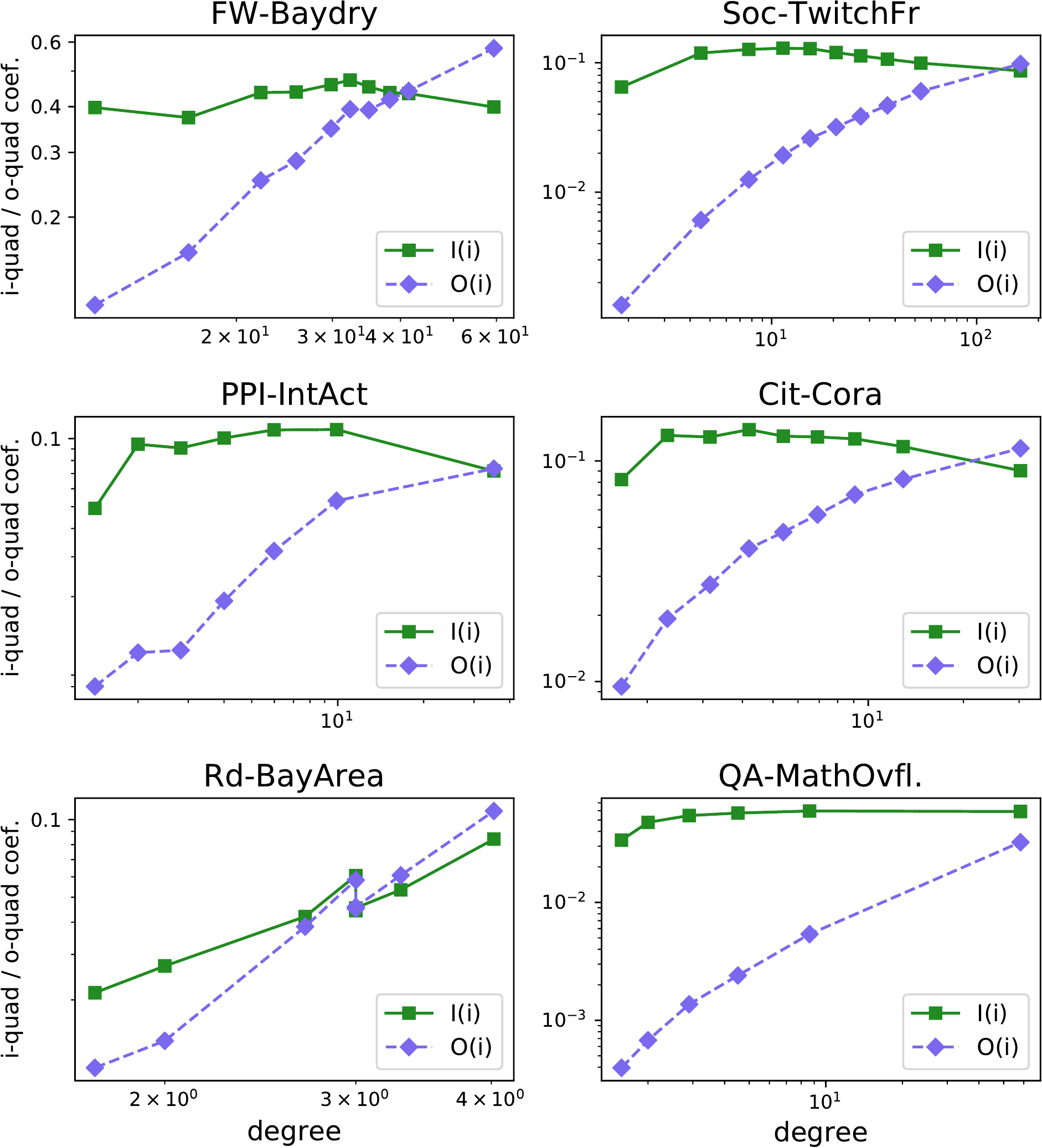}}
\caption{Correlation of two quadrangle coefficients with node degree in six real-world networks. Nodes are grouped into logarithmic bins in ascending order by degree, then average i-quad and o-quad coefficients are calculated in each bin.}
\label{fig:corr_degree}
\vspace{-0mm}
\end{figure}

\begin{prop} \label{prop}
Let $V$ be a set of $n$ nodes with specific degrees $d_1, d_2, ..., d_n$, on which graph $G$ is generated from the configuration model. Let $m=\frac{1}{2} \sum_{i=1}^{n} d_{i}$ denote the number of edges and $\bar{k}=(\sum_{i} d_{i}^{2}) /(\sum_{i} d_{i})$ be the expected degree when a node is chosen with probability proportional to its degree. As $n \rightarrow \infty$, for any node $i \in V$, its local i-quad coefficient satisfies:
\begin{equation*}
\mathbb{E}[I(i)]=\frac{(\bar{k}-1)^2}{2 m},
\end{equation*}
and its local o-quad coefficient satisfies:
\begin{equation*}
\mathbb{E}[O(i)]=\frac{(d_{i}-1) \cdot (\bar{k}-1)}{2 m} .
\end{equation*}
\end{prop}

\begin{proof}
For any open quadriad with node $i$ as an inner node, we denote one outer node by $k$ and another outer node by $l$ (Figure~\ref{fig:proof}a). The probability that this open quadriad is closed equals the probability of having an edge between node $k$ and $l$, which is $\left(d_{k}-1\right)\left(d_{l}-1\right) / 2 m$ in the configuration mode. The reason of subtracting $1$ from $d_k$ and $d_l$ is that one stub of node $k$ (and node $l$) has already been used in forming the open quadriad.

Now, we show that as $n \rightarrow \infty$, $\mathbb{E}\left[d_{k}\right]=\mathbb{E}\left[d_{l}\right]=\bar{k}$. Via stub matching, any node, other than node $i$ and $j$, can form an edge with node $j$ and thus become one outer node of the open quadriad. The probability of node $k$ being this node is proportional to its degree, which is $\frac{d_{k}}{\sum_{k\in V, k \neq i, j} d_{k}}$. Therefore, we have $\mathbb{E}\left[d_{k}\right]=\sum_{k \in V, k \neq i, j} d_{k} \cdot \frac{d_{k}}{\sum_{k\in V, k \neq i, j} d_{k}}$. When $n \rightarrow \infty$, $\mathbb{E}\left[d_{k}\right]=\sum_{k \in V} d_{k} \cdot \frac{d_{k}}{\sum_{k\in V} d_{k}}=\bar{k}$. Similarly, we have $\mathbb{E}\left[d_{l}\right]=\bar{k}$.

In short, we have:
\begin{equation*}\begin{aligned}
\mathbb{E}[I(i)] &=\mathbb{E}\left[\left(d_{k}-1\right)\left(d_{l}-1\right) /(2 m)\right] \\
& = \frac{(\mathbb{E}\left[d_{k}\right]-1) \cdot (\mathbb{E}\left[d_{l}\right]-1)}{2 m} =\frac{(\bar{k}-1) ^ 2}{2 m}.
\end{aligned}\end{equation*}

Likewise, for any open quadriad with node $i$ as an outer node, we denote the other outer node by $l$ (Figure~\ref{fig:proof}b). 
And we have:
\begin{equation*}\begin{aligned}
\mathbb{E}[O(i)] &=\mathbb{E}\left[\left(d_{i}-1\right)\left(d_{l}-1\right) /(2 m)\right] \\
& = \frac{(d_{i}-1) \cdot (\mathbb{E}\left[d_{l}\right]-1)}{2 m} =\frac{(d_{i}-1) \cdot (\bar{k}-1)}{2 m}.
\end{aligned}\end{equation*}
\end{proof}

\begin{figure}[t]
\centerline{\includegraphics[scale = 0.75]{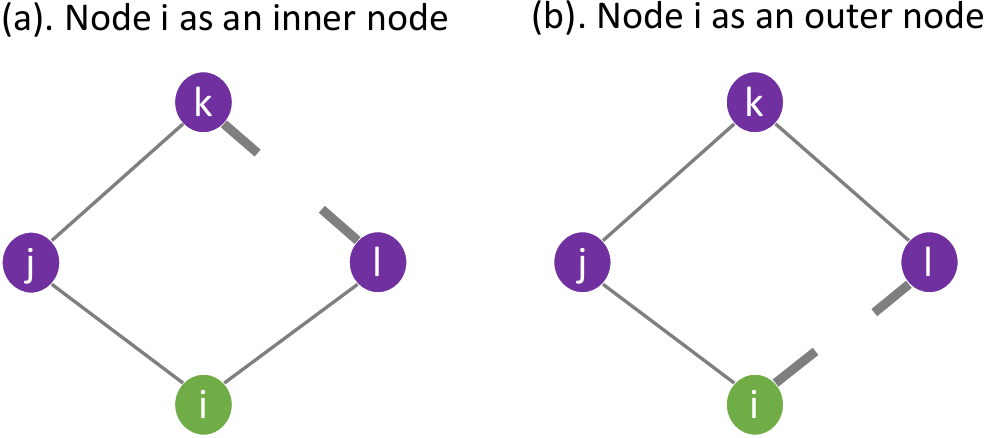}}
\caption{Two types of quadrangle formation via stub matching. (a) Quadrangle is potentially formed with the focal node $i$ acting as the inner node. The closing edge is between node $k$ and $l$. (b) Quadrangle is potentially formed with the focal node $i$ acting as the outer node. The closing edge is between node $i$ and $l$.}
\label{fig:proof}
\vspace{-0mm}
\end{figure}

Although Proposition~\ref{prop} is given under the configuration model, we see from Figure~\ref{fig:corr_degree} that this property is well preserved in most real-world networks. Only that in road networks, i.e., \textsc{Rd-NewYork} and \textsc{Rd-BayArea}, the average i-quad coefficient and the average o-quad coefficient are very similar (Table~\ref{tab:dataset}), and they exhibit similar correlations with node degree. This is because the variance of node degree is extremely small (less than one) in this type of network, resulting in $d_i$ close to $\bar{k}$, and thus $\mathbb{E}[O(i)]$ close to $\mathbb{E}[I(i)]$.

\subsection{Network classification}
In this section, we exhibit how useful the proposed quadrangle coefficients are in classifying different types of networks. Previous works have shown that normalized number of triads and triangles (triad significance profile\cite{milo2004superfamilies} and clustering signatures\cite{ahnert2008clustering}) are effective attributes in a network classification task. It motivated us to use the two quadrangle coefficients in the network classification, as they represent a normalized number of quadrangles.

We can see in Table~\ref{tab:dataset} that the quotient of the average i-quad coefficient and the average clustering coefficient ($\overline{I}/\overline{C}$), and the quotient of the average o-quad coefficient and the average closure coefficient ($\overline{O}/\overline{E}$) are contrasting in different types of networks. It is intuitive to expect the two quadrangle coefficients will be able to add useful discriminative information to a set of features, in addition to the average clustering coefficient and the average closure coefficient, for improving of the network classification accuracy.

\begin{figure*}[t]
\centerline{\includegraphics[scale = 0.42]{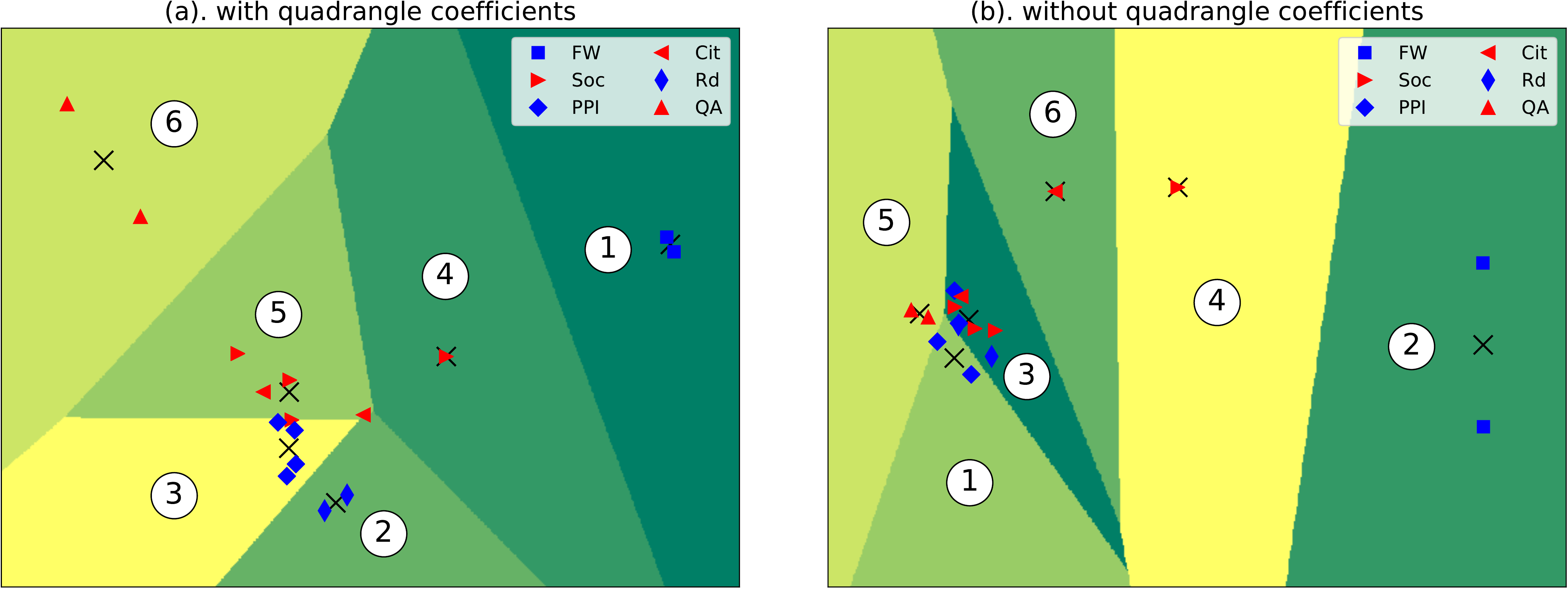}}
\caption{Two-dimensional visualisation of K-means clustering on PCA-reduced data, with and without quadrangle coefficients (left figure and right figure respectively). Six clusters are labelled from $1$ to $6$, and painted in sequential colours. Centroids of clusters are marked as black crosses. Data points are plotted in different shapes and colours representing their ground truth categories, as shown in the legend.}
\label{fig:classification}
\vspace{-0mm}
\end{figure*}

\noindent\textbf{\textit{Setup}.} 
We first prepare the data by choosing five features from the networks, i.e., the average node degree $\langle k\rangle$, the average clustering coefficient $\overline{C}$, the average closure coefficient $\overline{E}$, the average i-quad coefficient $\overline{I}$, and the average o-quad coefficient $\overline{O}$. Then we employ a K-means clustering algorithm to partition all $16$ networks of our dataset into $6$ clusters. The initial centroids are chosen randomly, and we repeat the algorithm with different set of initial centroids for $1000$ times, returning the best result in terms of V-measure score\cite{rosenberg2007v}. Maximum number of iterations for a single run is set to $300$. To compare, we use the same setting to run the experiment, but with only three features, i.e., without the two quadrangle coefficients.

\begin{table}[ht]
\setlength{\tabcolsep}{7pt} 
\centering
\caption{Homogeneity (Homo.), completeness (Compl.) and V-measure score of the K-means clustering on $16$ real-world networks, with and without the quadrangle coefficients (first row and second row respectively).}
\label{tab:network_classification}
\begin{tabular}{lccc}
\toprule
Features  & Homo. & Compl. & V-measure \\
\midrule
with quadrangle coefs.  &  0.810   &   0.879     & 0.826   \Tstrut \\
without quadrangle coefs. &  0.731  &   0.766     & 0.745  \Tstrut \\
\bottomrule
\end{tabular}
\end{table}

\noindent\textbf{\textit{Results and discussion}.}
The classification results measured in homogeneity, completeness and V-measure score are given in Table~\ref{tab:network_classification}. Homogeneity measures whether the samples of a single class belonging to a single cluster; Completeness measures whether all members of a class are assigned to the same cluster; V-measure score is the harmonic mean of the Homogeneity and Completeness. We observe significant improvement (more than $10\%$ in homogeneity and V-measure score, nearly $15\%$ in completeness score) after adding the two quadrangle coefficients. In order to better analyse the results, we adopt the Principal Component Analysis algorithm to compress the data to a two-dimensional space, and thus visualise the classification results (Figure~\ref{fig:classification}). 

We can see from Figure~\ref{fig:classification}(a) that with the two quadrangle coefficients, the labellings of food webs (cluster $1$), PPI networks (cluster $3$), road networks (cluster $2$) and QA networks (cluster $6$) are perfect. The model only cannot properly partition social networks from citation networks (cluster $5$ contains three social networks and two citation networks while cluster $4$ has only one social network --- \textsc{Soc-EmailEu}). In contrast, when the quadrangle coefficients are excluded from the model, the majority of data points are congregated at the left part of the space, resulting in worse classification result, as shown in Figure~\ref{fig:classification}(b). Only two types of networks are labelled perfectly (food webs in cluster $2$ and QA networks in cluster $5$). The remaining four types of networks are poorly clustered, especially in cluster $3$ which contains data points of all four categories. This experiment shows that adding quadrangle coefficients improves significantly the ability to tell apart different types of real-world networks, especially for these rich in quadrangles.  


\subsection{Link prediction}

As two new metrics measuring quadrangle formation, the i-quad coefficient and the o-quad coefficient provide additional topological features for a node-level network analysis and inference. As an example, we show their utilities in missing link prediction, where significant improvement is brought by adding them.

Many studies have shown that common neighbours index and its variations such as Adamic-Adar index and resource allocation index perform well in the link prediction problem \cite{liben2007link, adamic2003friends, zhou2009predicting}. Besides, the clustering coefficient and the closure coefficient are proven to be useful features to improve the performance \cite{al2006link, yin2019local}. Therefore, we use these five features as the baseline features in our prediction model, and then test the performance by adding the proposed i-quad and o-quad coefficients. XGBoost, the gradient boosted trees, is used as the prediction model due to its speed and performance. 

\noindent\textbf{\textit{Setup}.} 
We model a network as a graph $G = (V,E)$. For networks having timestamps on edges, we order the edges according to their appearing times and select the first $70\%$ edges and related nodes to form an \say{old graph}, denoted $G_{old} = (V^*, E_{old})$. For networks not having timestamps, we randomly shuffle the edges then perform the partition, and we repeat $100$ times in order to assess variance and reduce the impact of a single partition on the possible conclusions. The remaining $30\%$ edges filtered by node set $V^*$ will form a \say{new graph}, denoted $G_{new} = (V^*, E_{new})$. The test set is built by node pairs, that appear in the old graph, but do not form a link. Each such pair of nodes indicates a positive or a negative example depending on whether a link between them appears in the new graph. 

The training set is built on the old graph, on which we fit four XGBoost models with four sets of features: 1) baseline feature set which includes common neighbours, Adamic-Adar, resource allocation, clustering coefficient and closure coefficient; 2) baseline features plus i-quad coefficient; 3) baseline features plus o-quad coefficient; 4) baseline features plus both i-quad coefficient and o-quad coefficients. Then we evaluate their prediction performances on the test set. For large networks ($|V| > 10K$), we perform a randomised breadth first search sampling \cite{doerr2013metric} of $3K$ nodes on the original graph and repeat 10 times.

\begin{table}[t!]
\vspace{0mm}
\centering
\caption{Test set performance comparison measured in ROC-AUC score of four XGBoost models with different features. Second column lists the scores with baseline features (BL) , third column adds i-quad coefficient to baseline features, fourth column adds o-quad coefficient to baseline features, and fifth column adds both i-quad and o-quad coefficients to baseline features. An improvement of more than $2\%$ is put in bold type, and an improvement of more than $5\%$ is indicated by dagger. Last row gives the average (over the datasets) ranking of the four models for comparison, where smaller is better. A model receives rank $1$ if it has the highest ROC-AUC score, rank $2$ if it has the second highest, and so on. If two models share the best score, they both get rank $1.5$, and so on. The best ranking is put in bold italic.} 
\label{tab:link_pred_roc}
\setlength{\tabcolsep}{4.5pt} 
\def\arraystretch{1.2}
\begin{tabular}{lcccc}
\toprule
Network   & \multicolumn{1}{c}{\begin{tabular}[c]{@{}c@{}}w/ baseline\\ features (BL) \end{tabular}}   & \multicolumn{1}{c}{\begin{tabular}[c]{@{}c@{}}add I(i)\\ to BL\end{tabular}}   & \multicolumn{1}{c}{\begin{tabular}[c]{@{}c@{}}add O(i)\\ to BL\end{tabular}} & \multicolumn{1}{c}{\begin{tabular}[c]{@{}c@{}}add I(i) \& \\ O(i) to BL\end{tabular}} \\
\midrule
\textsc{FW-FloridaDry}  & 0.6703  & 0.6779 & 0.6834  & \textbf{0.6886} \\
\textsc{FW-LittleRock}  & 0.8077  & \textbf{0.8357}  & \textbf{0.8421}  & \textbf{0.8521$^\dagger$} \\
\midrule
\textsc{Soc-EmailEu$^{\tau}$}   & 0.9076  &  0.9070 &  0.9090  & 0.9084 \\
\textsc{Soc-ClgMsg$^{\tau}$}  & 0.7831  &  0.7873 &  0.7879   & 0.7920 \\
\textsc{Soc-BTCAlpha$^{\tau}$}  & 0.8588  & 0.8601   & 0.8679  & 0.8697 \\
\textsc{Soc-TwitchFr}  & 0.9160   & 0.9176  & 0.9192 & 0.9202 \\
\midrule
\textsc{PPI-Stelzl}  & 0.6565  & \textbf{0.7778$^\dagger$}  & \textbf{0.7809$^\dagger$} & \textbf{0.7764$^\dagger$} \\
\textsc{PPI-Figeys}  & 0.8171  & \textbf{0.8644$^\dagger$}  & \textbf{0.8668$^\dagger$}  & \textbf{0.8650$^\dagger$}\\
\textsc{PPI-Vidal}   & 0.7566  & \textbf{0.7973$^\dagger$} & \textbf{0.8009$^\dagger$}  & \textbf{0.7992$^\dagger$}\\
\textsc{PPI-IntAct}  & 0.8524  & \textbf{0.8808} & \textbf{0.8839} & \textbf{0.8842}\\
 \midrule
\textsc{Cit-DBLP$^{\tau}$}  & 0.7294 & 0.7261 & 0.7336 & 0.7310 \\
\textsc{Cit-Cora}    & 0.8700  & 0.8705 & 0.8726 & 0.8734 \\
 \midrule
\textsc{Rd-NewYork}  & 0.5268 & \textbf{0.5529} & \textbf{0.5538$^\dagger$} & \textbf{0.5538$^\dagger$} \\
\textsc{Rd-BayAera}  & 0.5218  & \textbf{0.5353} & \textbf{0.5353} & \textbf{0.5356}  \\
\midrule
\textsc{QA-MathOvfl.$^{\tau}$}  & 0.8546 & 0.8554  & 0.8541  & 0.8551 \\
\textsc{QA-AskUbuntu$^{\tau}$}  & 0.8746 & 0.8791 & 0.8765  & 0.8777 \\
\midrule
\midrule
\textbf{Avg. ranking}  &  3.8 & 2.8 & 1.9 & \textit{\textbf{1.5}} \\
\bottomrule
\end{tabular}
\vspace{-1mm}
\end{table}

\noindent\textbf{\textit{Results and discussion}.} Since a network link prediction is a highly unbalanced task, we choose ROC-AUC score as the metric and report the prediction result on the test set, as shown in Table~\ref{tab:link_pred_roc}. First, we discover that adding only i-quad (\nth{3} column) or o-quad coefficient (\nth{4} column) leads to improvement in most networks (except \textsc{Soc-EmailEu} and \textsc{Cit-DBLP} when i-quad is added, and \textsc{QA-MathOvfl.} when o-quad is added). And when both of the quadrangle coefficients are added to the baseline features (\nth{5} column), the performance is improved in all networks. The average ranking (last row) also shows that adding both i-quad and o-quad coefficients at the same time leads to the best performance overall. 

Second, we find that the improvement is particularly significant in food webs, protein-protein networks and road networks (more than $2\%$ in all eight networks of these three types, and more than $5\%$ in five networks when both quadrangle coefficients are added). The common characteristic of these types of networks is that they tend to have larger quadrangle coefficients compared to the clustering and closure coefficients. Also, we notice that only adding o-quad coefficient has better performance than only adding i-quad coefficient in most networks (except in two Q\&A networks), which is an interesting phenomenon for further study.


\section{Additional Related Work} \label{sec: related_works}

\begin{figure}[t]
\centerline{\includegraphics[scale = 0.75]{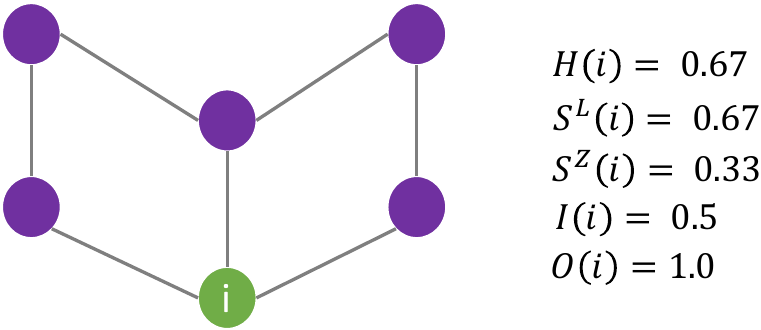}}
\caption{An example of all five coefficients measuring quadrangle formation for node $i$. $H(i)$ is the higher order clustering coefficient proposed by Fronczak et al.\cite{fronczak2002higher}; $S^L(i)$ is the square clustering coefficient proposed by Lind et al.\cite{lind2005cycles}; $S^Z(i)$ is another square clustering coefficient proposed by Zhang et al.\cite{zhang2008clustering}; $I(i)$ and $O(i)$ are the two quadrangle coefficients proposed by us.}
\label{fig:related_work}
\vspace{-0mm}
\end{figure}

We now recapitulate some additional related works that proposed other metrics to measure quadrangle formations in networks. Fronczak et al.\cite{fronczak2002higher} proposed a higher order clustering coefficient for random networks. It is defined as $C_{i}(x)=\frac{2 E_{i}(x)}{k_{i}\left(k_{i}-1\right)}$, where $i$ is the focal node and $x$ is the length of path. $E_{i}(x)$ denotes the number of $x$-length paths between the neighbours of $i$. When $x$ equals $2$, this definition deals with the formation of quadrangles. The limitation of this definition is that the normalisation only takes the degree of the focal node $i$ into account while neglects the degree of $i$'s neighbours. Since each pair of neighbours could have multiple length-$2$ paths between them, the clustering value can be larger than one.

Lind et al.\cite{lind2005cycles} later proposed a square clustering coefficient in the context of bipartite networks by taking into consideration the degree of the neighbours, in other words, the length-$2$ paths starting from the focal node.  It is defined as $C_{4, m n}(i)=\frac{q_{i m n}}{\left(k_{m}-\eta_{i m n}\right)\left(k_{n}-\eta_{i m n}\right)+q_{i m n}}$, where $m$ and $n$ are a pair of neighbours of the focal node $i$, and 
$q_{i m n}$ denotes the number of squares containing the three nodes. What is uncommon about this definition is that it deems squares are formed via node overlapping, which is not a standard approach. Zhang et al.\cite{zhang2008clustering} then modified the equation and proposed another more standard square clustering coefficient for bipartite networks. Their definition is:  $C_{4, m n}(i)=\frac{q_{i m n}}{\left(k_{m}-\eta_{i m n}\right)+\left(k_{n}-\eta_{i m n}\right)+q_{i m n}}$. However, in both of these definitions, there is no notion of open quadriad introduced, and the normalisation is thus based on the number of squares.

Our proposed i-quad and o-quad coefficients are different from the previous works in that 1) the scope of the o-quad coefficient is larger since it takes into account length-$3$ paths emanating from the focal node, whereas the square clustering coefficients only calculates length-$2$ paths in the normalisation; 2) the two coefficients proposed by us view a formed quadrangle as being built from two open quadriads, which conform with the classic clustering and closure coefficients (in their definitions a formed triangle is viewed as being built from open triads); 3) the two quadrangle coefficients are proposed for the general unipartite networks on which multiple experiments are conducted. In Figure~\ref{fig:related_work}, we provide a small example to illustrate the three coefficients proposed by previous works and the two quadrangle coefficients proposed by us.

\section{Conclusion} \label{sec: conclusion}
In this paper, we introduced the i-quad coefficient and the o-quad coefficient to measure quadrangle formation in networks, according to the different location of the focal node in an open quadriad. We also extended them to weighed networks.  Through experiments on $16$ real-world networks from six domains, we revealed that 1) in most types of networks, the average o-quad coefficient is smaller than the average i-quad coefficient; 2) in food webs, protein-protein interaction networks and road networks, the i-quad and o-quad coefficients are larger than the clustering and closure coefficients respectively; 3) the o-quad coefficient tends to increase with node degree while the i-quad coefficient does not change too much as the node degree increases.

We also demonstrated that including the two coefficients leads to improvement in both network-level and node-level analysis tasks, such as network classification and link prediction. The improvement is especially significant in food webs, protein-protein networks and road networks in link prediction task. 
Additionally, we plan to further consider the dynamics of time-varying networks and link directions of directed networks when measuring quadrangle formation in the future. Due to the simplicity and interpretability in the definitions, we anticipate that the i-quad and o-quad coefficients will become standard descriptive features and be incorporated in other network mining tasks.

\ifCLASSOPTIONcompsoc
  \section*{Acknowledgments}
This work was supported by the Australian Research Council, grant No. DP190101087: \say{Dynamics and Control of Complex Social Networks}.
\else
  \section*{Acknowledgment}
\fi

\bibliographystyle{IEEEtran}
\bibliography{Reference}

\ifCLASSOPTIONcaptionsoff
  \newpage
\fi

\end{document}